\documentclass{article}

\usepackage[utf8]{inputenc}
\usepackage{amsmath,amssymb,amsfonts}
\usepackage{color}
\usepackage{latexsym}
\usepackage{repsize}
\usepackage{epsfig}
\usepackage{tikz, pgfplots}
\pgfplotsset{compat=newest}
\usetikzlibrary{positioning}
\usetikzlibrary{decorations.pathreplacing}

\newtheorem{theorem}{Theorem}
\newtheorem{lemma}{Lemma}
\newenvironment{proof}{\noindent{\bf Proof:}}{\nobreak\hfill\nobreak$\Box$\par\medskip}

\newcommand{\remove}[1]{\relax}

\newcommand{\defeq}{\stackrel{\text{def}}{=}}

\newcommand{\ld}[1]{\ensuremath{{\rm ld}(#1)}}
\newcommand{\OPT}{\ensuremath{{\sf OPT}}}
\newcommand{\OPTs}{\ensuremath{\OPT(\sigma)}}
\newcommand{\OPTsa}{\ensuremath{\OPT(\sigma_{\alpha})}}
\newcommand{\DNF}{\ensuremath{{\sf DNF}}}
\newcommand{\DHk}{\ensuremath{{\sf DH}_k}}

\newcommand{\DHak}{\ensuremath{{\sf DH}^{b+\lg n}_k}}
\newcommand{\DHab}{\ensuremath{{\sf DH}^{b+\lg n}_2}}
\newcommand{\DHac}{\ensuremath{{\sf DH}^{b+\lg n}_3}}
\newcommand{\DHad}{\ensuremath{{\sf DH}^{b+\lg n}_4}}

\newcommand{\G}{\ensuremath{{\cal G}}}
\newcommand{\Gt}[1]{\ensuremath{\G_{{\hspace*{-0.05em}#1}}}}

\begin{document}
  \title{Online Bin Covering with Exact Parameter Advice}
   \author{Andrej Brodnik\footnotemark[1], Bengt J.~Nilsson\footnotemark[2], and Gordana Vujovic\footnotemark[3] \vspace*{1ex}\\
      \footnotemark[1] University of Ljubljana, Slovenia
      ({\tt andrej.brodnik@fri.uni-lj.si}) \\
      \footnotemark[2] Malmö University, Sweden 
      ({\tt bengt.nilsson.TS@mau.se}) \\
      \footnotemark[3] University of Ljubljana, Slovenia
      ({\tt gogili.vujovic@gmail.com})
    }
   \date{}
      
  %\titleodd{Online Bin Covering with Exact Parameter Advice}
  %\authoreven{A.~Brodnik~{\em et~al}.}
  %\keywords{Online computation, Competitive analysis, Advice complexity, Bin covering}
  %\received{April 3, 2023}

  \maketitle
  
%\vspace*{-6ex}
\begin{abstract}
We show an asymptotic $2/3$-competitive strategy for the bin covering problem using $O(b+\log n)$ bits of advice, where $b$ is the number of bits used to encode a rational value and $n$ is the length of the input sequence. 
\end{abstract}
  %\abstractSi{V članku je predstavljena asimptotično $2/3$-konkurenčna strategija za problem pokrivanja košev z uporabo $O(b+\log n)$ bitnega nasveta, kjer je $b$ število bitov, uporabljenih za kodiranje racionalne vrednosti, $n$ pa je dolžina vhodnega niza.}

\section{Introduction}\label{sec:intro}\vspace*{-2ex}
In the bin covering problem, we are given a set of items of different sizes in the range $]0,1]$ and the goal is to find a maximum number of covered bins where a bin is covered if the sizes of items placed in it is at least~$1$.
It has been shown that the bin covering problem is NP-hard~\cite{ASSMANN1984:bincovering}.
The covering problem has applications in various situations in business and in industry, from packing snack pieces into boxes so that each box contains at least its defined net weight, to such complex problems as redistribution tasks/items to a maximum number of factories/bins, all working at or beyond the minimal feasible level. 

The bin covering problem was studied in-depth in Assmann's Ph.D.~thesis~\cite{assmann1983problems}. In the online version, items are delivered successively (one-by-one) and each item has to be packed, either in an existing bin or a new bin, before the next item arrives. The quality of online strategies is measured by their {\em competitive ratio}, the minimum ratio between the quality of the strategy's solution and that of an optimal one. The first known online strategy that has been proposed for the problem is {\em Dual Next Fit\/} (\DNF), analogous to Next Fit for the bin packing problem. The competitive ratio of \DNF\ is $1/2$ as proved by Assmann~{\em et~al}.~\cite{ASSMANN1984:bincovering}. Csirik and Totik~\cite{CSIRIK1988163} prove that no online algorithm can achieve a competitive ratio better than~$1/2$. Further lower bounds are given by Balogh~{\em et~al}.~\cite{balogh2019}.

Thus, the only way to improve the competitive ratio is to change the computational model. We argue that the online framework is too restrictive in that it allows an all-powerful adversary to construct the input sequence in the worst possible way for the strategy making a competitive ratio better than~$1/2$ unattainable. 
Boyar~{\em et~al}.~\cite{boyar2021} look at bin covering using extra advice provided by an oracle through an advice tape that the strategy can read. If the input sequence consists of $n$ items, they show that with $o(\log\log n)$ bits of advice, no strategy can have better competitive ratio than $1/2$. In addition, they show that a linear number of bits of advice is necessary to achieve competitive ratio greater than~$15/16$. They also provide a strategy with $O(\log\log n)$ bits of advice having competitive ratio~$8/15$.

We note that in order to  provide exact advice values, e.g., an integer bounded by the number $n$, $\Omega(\log n)$ bits are required. Hence, the strategy presented by Boyar~{\em et~al} only provides approximate values of certain key parameters to achieve the $8/15$ competitive ratio. One can therefore ask if the competitive ratio can be improved if the oracle can give exact parameter values as advice. We answer this question affirmatively.

\subsection{Our Result}

We show an asymptotic $2/3$-competitive strategy for the bin covering problem using $O(b+\log n)$ advice, where $b$ is the number of bits used to encode a rational value in the input sequence and $n$ is the length of the input sequence. 
\section{Preliminaries}

The %{\em online vector covering problem\/} 
{\em online bin covering problem\/} 
we consider is, given an input sequence $\sigma=(v_1,v_2,\ldots)$, of rational values $v_i\in[0,1]$, find the {\em maximum\/} number of unit sized %$D$-dimensional
 bins that can be covered online with items from the input sequence $\sigma$. %From this problem definition we have that $0\leq v^{(i)}_j\leq1$, for all $i$ and $j$, i.e., all coordinates are non-negative. A {\em bin\/} here is simply a subset of the vectors in~$\sigma$. 
%The bin covering problem is a dual version of the {\em bin packing problem}.

We define %the {\em min norm} of a vector to be 
%\begin{equation}
%\|v\|_{\min}\defeq\min\{v_1,\ldots,v_D\}
%\end{equation}
%and define 
the {\em load\/} of a bin $B$ to be
%\begin{equation}
$\ld{B} \defeq \sum_{v\in B}v$.
%\end{equation}
We can similarly define the load of a sequence $\sigma$ to be $\ld{\sigma} \defeq \sum_{v\in\sigma}v$.

A {\em covering\/} is a partitioning of the items into bins $B_1,B_2,\ldots$ such that for each bin $B_j$%\vspace*{-2ex}
\begin{align}\label{eqn:feasible}
%\sum_{v\in B_j}v\geq1
\ld{B_j}\geq1
\end{align}
and our objective is to find the maximum number of bins that satisfy Inequality~(\ref{eqn:feasible}). In contrast to the bin packing problem, a strategy can open any number of bins at any time. However, only those that are filled to a load of at least 1 are counted in the solution.

We measure the quality of an online maximization strategy by its {\em competitive ratio}, the maximum bound $R$ such that 
%\begin{equation}
$\big|{\sf A}(\sigma)\big|\geq R\cdot\big|\OPT(\sigma)\big|-C$,
%\end{equation} 
for every possible input sequence $\sigma$, where ${\sf A}(\sigma)$ is the solution produced by the strategy ${\sf A}$ on $\sigma$, $\OPT(\sigma)$ is a solution on $\sigma$ for which $|\OPT(\sigma)|$ is maximal, and $C$ is some constant. If $C=0$, we say that the competitive ratio is {\em strict}, otherwise it is~{\em asymptotic}.
The competitive ratio $R$ is thus a positive real value~$\leq1$, where equality to $1$ implies that the strategy is (asymptotically) optimal.

Of particular interest is the Dual Next Fit strategy (\DNF), where \DNF\ maintains one active bin $B$, and packs the items into $B$ until it is covered. It then opens a new empty bin as the active bin and continues the process. As said, Assmann~{\em et~al}.~\cite{ASSMANN1984:bincovering} prove that \DNF\ has a competitive ratio of $1/2$ and Csirik and Totik~\cite{CSIRIK1988163} prove that no online algorithm can achieve a competitive ratio better than~$1/2$.

If we know some further structure of the input sequence, we can do slightly better as is shown in the next lemma that we will make extensive use of in the sequel.

\begin{lemma}\label{lem:dnf}
The online strategy\/ \DNF\ for the  bin covering problem on an input sequence $\sigma_{\alpha}$ where the items have weights bounded by $\alpha<1$ has cost
$$
\big|\DNF(\sigma_{\alpha})\big| > \frac{1}{1+\alpha}\big|\OPTsa\big| - \frac{1}{1+\alpha}.
$$
\end{lemma}
\begin{proof}
Assume that \DNF\ opens $s+1$ bins when accessing the sequence $\sigma_{\alpha}$, $s$ of which are covered. Since every item has weight at most $\alpha$, it means that each of the $s$ covered bins are filled to a total weight of less than $1+\alpha$. A bin not obeying this limit would have been covered already before \DNF\ places the last item in it, a contradiction. Thus the total load of the sequence $\sigma$ is
\begin{equation*}
(1+\alpha)s + 1 > \ld{\sigma_{\alpha}} \geq \big\lfloor\ld{\sigma_{\alpha}}\big\rfloor\geq\big|\OPTsa\big|,
\end{equation*}
whereby 
$|\DNF(\sigma_{\alpha})|=s>|\OPTsa|/(1+\alpha) - 1/(1+\alpha)$ as claimed.
\end{proof}

Another strategy of interest is Dual Harmonic (\DHk), where the strategy subdivides the items by sizes into $k$ groups,
\begin{equation*}
]0,1/k[,~[1/k,1/(k-1)[,\ldots,[1/3,1/2[,~[1/2,1[,
\end{equation*}
and packs items in each group, maintaining $k$ groups, according to~\DNF. Evidently, \DHk\ is at best $1/2$-competitive using the same argument as in Csirik and Totik~\cite{CSIRIK1988163}.

In certain situations, the complete lack of information about future input is too restrictive. In a sense, the online strategy plays a game against an all-powerful adversary who can construct the input sequence in the worst possible manner. To alleviate the adversary's advantage, we consider the following {\em advice-on-tape\/} model~\cite{bockenhauer2009advice:advice}. An {\em oracle\/} has knowledge about both the strategy and the full input sequence from the adversary, it writes information on an {\em advice tape\/} of unbounded length. The strategy can read bits from the advice tape at any time, before or while the requests are released by the adversary. The {\em advice complexity\/} is the number of bits read from the advice tape by the strategy. Since the length of the advice bit string is not explicitly given, the oracle is unable to encode information into the length of the string, thereby requiring some mechanism to infer how many bits of advice the strategy should read at each step. This can be done with a self-delimiting encoding that extends the length of the bit string only by an additive lower order term~\cite{Boyaretal:advice}.

A bit string $s$ is encoded as $e(s)=u(s)\circ b(s)\circ s$ ($\circ$ denotes concatenation), where $b(s)$ is a binary encoding of the length of the string $s$ and $u(s)$ consists of $\big|b(s)\big|$ ones followed by a single zero, thus indicating how many bits the strategy needs to read in order to obtain the length of the string $s$. The encoding has length at most $\big|e(s)\big|=|s|+2\lceil\log(|s|+1)\rceil+1$.
We henceforth assume that all advice information is encoded in this way. An integer $m$ can thus be encoded exactly using $O(\log m)$ bits and a rational value $m_e/m_d$, where $m_e$ and $m_d$ are integers can be encoded using $O(\log m_e + \log m_d)$ bits. If the rational value lies in the interval $[0,1]$, then $m_e\leq m_d$ and the encoding can be made using $O(\log m_d)$ bits.

We will base our strategy on \DHk\ with added advice to improve on the competitive ratio, as do Boyar~{\em et~al}.~\cite{boyar2021}.
%\section{Strategies for 1-Dimensional Bin Covering}\label{sec:1dim}
{

\newcommand{\Gttj}{\ensuremath{\G_{\hspace*{-0.05em}t_1 t_2\cdots t_j}}}
\newcommand{\Gttjp}{\ensuremath{\G_{\hspace*{-0.05em}t_1 t_2\cdots t_j t_{j+1}}}}
\newcommand{\Gtwotwo}{\ensuremath{\G_{\hspace*{-0.05em}22}}}
\newcommand{\Gtwo}{\ensuremath{\G_{\hspace*{-0.05em}2}}}
\newcommand{\GS}{\ensuremath{{\cal G}_{\hspace*{-0.05em}S}}}
\newcommand{\Gthree}{\ensuremath{\G_{\hspace*{-0.05em}3}}}
\newcommand{\Gtwothree}{\ensuremath{\G_{\hspace*{-0.05em}23}}}

\newcommand{\IS}[1]{\ensuremath{{\cal I}_{S}(#1)}}
\newcommand{\It}[2]{\ensuremath{{\cal I}_{#2}(#1)}}
\newcommand{\Itwo}[1]{\ensuremath{{\cal I}_{2}(#1)}}

%\section{An Exact Strategy for One-Dim\-en\-sion\-al Vector Covering}\label{sec:1dimx}
\section{Exact Advice-based Strategies for Bin Covering}\label{sec:1dimx}

Each item $v$ in the input sequence corresponds to a rational value $0<v<1$, since any $v$ above or equal to 1 will cover a bin and then the optimal solution can be assumed to place $v$ alone in a bin to cover it. Also, values of size 0 could be placed in the first covered bin without loss of generality.

Fix an integer $k\geq2$. We will subdivide the set of items into $k$ subsets, such that $1/t\leq v<1/(t-1)$ for each integer $2\leq t\leq k$, the {\em $t$-items}, and items $v<1/k$, the {\em small items}.

Consider a fixed optimal covering \OPTs\ for the input sequence $\sigma$. We can partition the solution \OPTs\ into groups, \Gttj, where the index $t_1 t_2\cdots t_j$, with $2\leq t_1\leq t_2\leq\cdots\leq t_j\leq k$, denotes that each bin in group \Gttj\ contains one $t_1$-item, one $t_2$-item, etc, multiplicity denoting the number of times each item type occurs in the bin. The group of bins that are only covered by small items is denoted by~\GS.

We say that a bin in group \Gttj\ is {\em easy}, if $\sum_{t\in\{t_1,t_2,\ldots,t_j\}}1/t\geq1$ and we can assume without loss of generality that easy bins contain no small items. 
Furthermore, we assume that if the bins in \Gttj\ are easy, then any bin group \Gttjp\ is empty, if $t_1 t_2\cdots t_j$ is a subsequence of $t_1 t_2\cdots t_{j+1}$, as the $t_{j+1}$-item in a bin in \Gttjp\ can be moved to other bins while we still maintain coverage in the bin. As an example, \Gtwotwo\ are those bins that each contain two $2$-items, so those bins are easy since two $2$-items together guarantee that the bin is covered and if a bin in an optimal solution contains two $2$-items, we assume that it does not contain any other items. As we noted, such items can be moved to other non-easy bins.

The non-easy bins may contain small items. Consider a bin in \Gtwothree, if $k\geq3$, that contain one $2$-item and one $3$-item. If the $2$-item is $0.6$ and the $3$-item is $0.4$, then the items cover the bin but this is not guaranteed since if the two items are $0.55$ and $0.35$, there must be small items in the bin for it to be covered. 

We also say that a bin in \Gttj\ is a {\em gap bin}, if $\sum_{t\in\{t_1,t_2,\ldots,t_j\}}1/(t-1)<1$, as each of these bins must contain small items to the amount of more than $1-\sum_{t\in\{t_1,t_2,\ldots,t_j\}}1/(t-1)$ to be covered. For example, the bins in \Gthree\ are gap bins since they all have to contain small items to an amount of more than~$1/2$ to be covered. 

The size of the optimal solution is given by
\begin{align}
|\OPTs| &= \sum_{\forall t_1 t_2\cdots t_j} |\Gttj| + |\GS|,
\end{align}
for all valid index combinations~$t_1 t_2\cdots t_j$.

We modify \DHk\ to operate on advice and describe this strategy, denoted \DHak\!, dependent on the parameter $k$, the number of item types used to partition the items into. The superscript indicates the amount of advice that the strategy admits. Let $x_1,\ldots,x_{n}$, $n=|\sigma|$, be an ordering of the items in $\sigma$, such that $x_i\geq x_{i+1}$, for $1\leq i<|\sigma|$. The oracle provides the strategy with an integer $m$ and the value $x_m$ through a self-delimiting encoding. 

The objective of the parameters $m$ and $x_m$ is to give the strategy enough information to emulate the construction of the bin group \Gtwo\ in an optimal solution, independently of the ordering in which items of the input sequence occur. The value of $m$ is balanced by the size $|\Gtwo|$, also designating the number of 2-items packed singly but together with small items in the bins of \Gtwo, and the amount of small items present in the bins of \Gtwo. In the optimal solution, each of these bins could be covered to exactly the value~$1$ but any online strategy may have to overfill by an amount~$<1/k$.
\begin{figure*}
\centering
\begin{tikzpicture}[scale=0.29]

\filldraw[thick, fill=blue!15!] (0.5,0) rectangle node{$0.53$} +(3.0, 5.3);
\filldraw[fill=blue!15!] (0.5,5.3) rectangle node{$0.51$} +(3.0, 5.1);

\filldraw[thick, fill=blue!15!] (4,0) rectangle node{$0.52$} +(3.0, 5.2);
\filldraw[fill=blue!15!] (4,5.2) rectangle node{$0.51$} +(3.0, 5.1);

\draw [decorate,decoration={brace,amplitude=5pt,mirror,raise=1ex}]
  (0.5,0.5) -- (7,0.5) node[midway, yshift=-1.5em]{$G_{22}$};
  
\filldraw[thick,fill=blue!15!] (7.5,0) rectangle node{$0.90$} +(3.0,9);
\filldraw[fill=red!15!] (7.5,9) rectangle node{$0.11$} +(3.0,1.1);

\filldraw[fill=blue!15!] (11,0) rectangle node{$0.80$} +(3.0,8);
\filldraw[fill=red!15!] (11,8) rectangle node{$0.20$} +(3.0,2);

\filldraw[fill=blue!15!] (14.5,0) rectangle node{$0.72$} +(3.0,7.2);
\filldraw[fill=red!15!] (14.5,7.2) rectangle node{$0.28$} +(3.0,2.8);

\filldraw[thick, fill=blue!15!] (18,0) rectangle node{$0.67$} +(3.0,6.7);
\filldraw[fill=red!15!] (18,6.7) rectangle node{$0.15$} +(3.0,1.5);
\filldraw[fill=red!15!] (18,8.2) rectangle node{$0.18$} +(3.0,1.8);

\filldraw[thick, fill=blue!15!] (21.5,0) rectangle node{$0.55$} +(3.0, 5.5);
\filldraw[fill=red!15!] (21.5,5.5) rectangle node{$0.30$} +(3.0,3);
\filldraw[fill=red!15!] (21.5,8.5) rectangle node{$0.15$} +(3.0,1.5);

\draw [decorate,decoration={brace,amplitude=5pt,mirror,raise=1ex}]
  (7.5,0.5) -- (24.5,0.5) node[midway, yshift=-1.5em]{$G_2$};

\filldraw[thick, fill=blue!15!] (25,0) rectangle node{$0.60$} +(3.0,6);
\filldraw[fill=green!15!] (25,6) rectangle node{$0.40$} +(3.0,4);
\draw [decorate,decoration={brace,amplitude=5pt,mirror,raise=1ex}]
  (25,0.5) -- (28,0.5) node[midway, yshift=-1.5em]{$G_{23}$};

\filldraw[thick, fill=green!15!] (28.5,0) rectangle node{$0.45$} +(3.0, 4.5);
\filldraw[fill=green!15!] (28.5,4.5) rectangle node{$0.35$} +(3.0, 3.5);
\filldraw[thick, fill=red!15!] (28.5,8) rectangle node{$0.15$} +(3.0, 1.5);
\filldraw[fill=red!15!] (28.5,9.5) rectangle node{$0.10$} +(3.0, 1.0);

\filldraw[thick, fill=green!15!] (32,0) rectangle node{$0.45$} +(3.0, 4.5);
\filldraw[fill=red!15!] (32,4.5) rectangle node{$0.25$} +(3.0, 2.5);
\filldraw[thick, fill=green!15!] (32,7) rectangle node{$0.35$} +(3.0, 3.5);
  
\filldraw[thick, fill=green!15!] (35.5,0) rectangle node{$0.42$} +(3.0, 4.2);
\filldraw[fill=green!15!] (35.5,4.2) rectangle node{$0.41$} +(3.0, 4.1);
\filldraw[thick, fill=red!15!] (35.5,8.3) rectangle node{$0.30$} +(3.0, 3.0);
\draw [decorate,decoration={brace,amplitude=5pt,mirror,raise=1ex}]
  (28.5,0.5) -- (38,0.5) node[midway, yshift=-1.5em]{$G_{33}$};

\filldraw[thick, fill=blue!15!] (0.5,-16) rectangle node{$0.80$} +(3.0, 8.8);
\filldraw[fill=red!15!] (0.5,-7.2) rectangle node{$0.25$} +(3.0, 2.5);

\filldraw[thick, fill=blue!15!] (4,-16) rectangle node{$0.90$} +(3.0, 9);
\filldraw[fill=red!15!] (4,-7) rectangle node{$0.20$} +(3.0, 2.8);
\draw [decorate,decoration={brace,amplitude=5pt,mirror,raise=0ex}]
  (0.5,-16) -- (7,-16) node[midway, yshift=-1em]{critical bins};
  
\filldraw[thick, fill=blue!15!] (8,-16) rectangle node{$0.72$} +(3.0, 7.2);
\filldraw[fill=blue!15!] (8,-8.8) rectangle node{$0.51$} +(3.0, 5.1);

\filldraw[thick, fill=blue!15!] (11.5,-16) rectangle node{$0.67$} +(3.0, 6.7);
\filldraw[fill=blue!15!] (11.5,-9.3) rectangle node{$0.60$} +(3.0, 6.0);

\filldraw[thick, fill=blue!15!] (15,-16) rectangle node{$0.55$} +(3.0, 5.5);
\filldraw[fill=blue!15!] (15,-10.5) rectangle node{$0.53$} +(3.0, 5.3);

\filldraw[thick, fill=blue!15!] (18.5,-16) rectangle node{$0.52$} +(3.0, 5.2);
\filldraw[fill=blue!15!] (18.5,-10.8) rectangle node{$0.51$} +(3.0, 5.1);

\filldraw[thick, fill=green!15!] (22,-16) rectangle node{$0.45$} +(3.0, 4.5);
\filldraw[fill=green!15!] (22,-11.5) rectangle node{$0.45$} +(3.0, 4.5);
\filldraw[fill=green!15!] (22,-7) rectangle node{$0.42$} +(3.0, 4.2);

\filldraw[thick, fill=green!15!] (25.5,-16) rectangle node{$0.41$} +(3.0, 4.1);
\filldraw[fill=green!15!] (25.5,-11.9) rectangle node{$0.35$} +(3.0, 3.5);
\filldraw[fill=green!15!] (25.5,-8.4) rectangle node{$0.35$} +(3.0, 3.5);

\filldraw[fill=red!15!] (29,-16) rectangle node{$0.28$} +(3.0, 2.5);
\filldraw[fill=red!15!] (29,-13.5) rectangle node{$0.11$} +(3.0, 1.1);
\filldraw[fill=red!15!] (29,-12.4) rectangle node{$0.15$} +(3.0, 1.5);
\filldraw[fill=red!15!] (29,-10.9) rectangle node{$0.15$} +(3.0, 1.5);
\filldraw[fill=red!15!] (29,-9.4) rectangle node{$0.15$} +(3.0, 1.5);
\filldraw[fill=red!15!] (29,-7.9) rectangle node{$0.10$} +(3.0, 1.0);
\filldraw[thick, fill=red!15!] (29,-6.9) rectangle node{$0.30$} +(3.0, 3);
\end{tikzpicture}
    \caption{\label{fig:example}(top) An optimal covering for instance $\sigma$ with~11 bins. (bottom) A \DHac\ covering with~9~bins.}
\end{figure*}
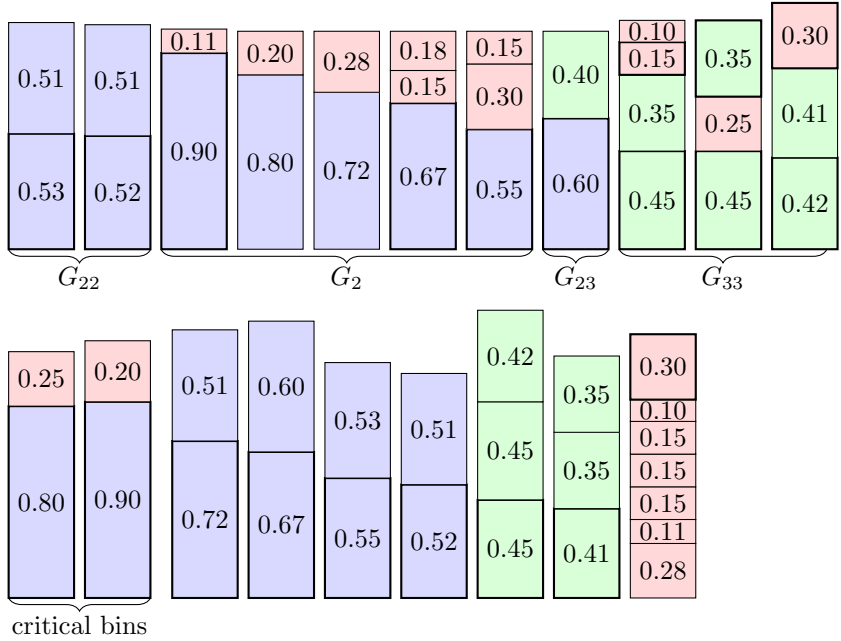

The oracle, given the knowledge of the input sequence and the strategy \DHak\!, can determine the best value for $m$ by emulating \DHak\ on the input sequence, for each integer $0\leq m\leq|\Gtwo|$, if it knows an optimal solution. If it does not, it counts the number of $2$-items, $n_2$, in the input sequence and tries all values between $0\leq m\leq n_2$ and reports an $m$ value for which \DHak\ delivers as large a solution as possible. Since $n_2\geq|\Gtwo|$ the best solution is no worse than if $m$ is restricted to be~$\leq|\Gtwo|$.
Given a specific integer $m$, the oracle can provide the value $x_m$ using the standard selection algorithm~\cite{BlumFloydPrattRivestTarjan}.

The strategy \DHak\ initially reads the parameters $m$ and $x_m$ and opens $m$ bins that we call {\em critical bins\/} and that will each be covered with one of the $m$ largest $2$-items of the input sequence $\sigma$ together with small items. Initially, each critical bin is assumed to have a {\em virtual load} of $x_m$. When an item of size $\geq x_m$ is placed in a critical bin,  its virtual load is increased to the actual value of the item. The strategy further opens a {\em $t$-bin\/} for every item type $t\in\{2,\ldots,k\}$, and a {\em small bin\/} for the small items. As the next item $v$ of the input sequence arrives, it is handled as follows:
\begin{enumerate}
\item
if $x_m\leq v$, place $v$ in the next critical bin that does not yet contain a $2$-item and update the virtual load of the critical bin,
\item
if $1/k\leq v<x_m$ is a $t$-item, place $v$ in the corresponding $t$-bin using \DNF. If the bin becomes covered, close it and open a new $t$-bin,
\item
if $v<1/k$ is small, place $v$ in the next critical bin that does not contain small items up to a virtual load of at least $1$ and update the virtual load of this critical bin. If all critical bins are filled up to a virtual load of 1, place $v$ in the small bin using \DNF. If the small bin becomes covered, close it and open a new small bin.
\end{enumerate} 

Consider the following example input sequence 
\begin{align*}
\sigma =\mbox{}& 
\big(0.25,0.80,0.72,0.20, 0.90, 0.45, 0.51, 
%\\&\
0.67, 0.45, 0.60, 0.42, 0.55, 0.53, 0.28, 0.11, 
\\&\
0.15, 0.52, 0.15, 0.51,0.41, 0.15, 0.35, 0.10,
%\\&\
0.35, 0.30, 0.30, 0.40, 0.18\big),
\end{align*}
 taken from~\cite{boyar2021} and slightly modified. We run it through \DHac\!, so $k=3$ given that the oracle provides the values $m=2$ and $x_m=0.80$. Figure~\ref{fig:example} shows the optimal solution and the \DHac\ solution on the sequence~$\sigma$.

We prove the following simple result to give the idea for the more extended version in Theorem~\ref{thm:twothirds}.
\begin{theorem}\label{thm:pointsix}
Assume that the strategy\/ \DHab\ has access to the exact values of $m$ and $x_m$, then it has asymptotic competitive ratio
$$
\big|\DHab(\sigma)\big|\geq\frac{3}{5}\big|\OPTs\big| - \frac{19}{15},
$$
for serving any sequence $\sigma$ of size~$n$.
\end{theorem}
\begin{proof}
Since our strategy uses only two item types, $2$-items and small items, the optimal solution consist of three bin groups with $|\OPTs|=|\Gtwotwo|+|\Gtwo|+|\GS|$.
Note that the number of 2-items in $\sigma$ is exactly 
\begin{align}\label{eq:twoeq}
T_2=2|\Gtwotwo| + |\Gtwo|.
\end{align}

Consider now some arbitrary set of covered bins \G, where each bin only contains small items. Assume that these bins have a total load of $S=\sum_{B\in\G}\ld{B}\geq|\G|$ and that the input sequence restricted to these small items is $\sigma_S$. From Lemma~\ref{lem:dnf} we have that 
\begin{equation}
\big|\DNF(\sigma_S)\big|  > \frac{2}{3}S - \frac{2}{3} \geq \frac{2}{3}\big|\G\big| - \frac{2}{3}
\end{equation}

We let $S_S\defeq\sum_{B\in\GS}\ld{B}$ be the total load of the small items covering the bins in~$\GS$.

Next, we analyze the competitive ratio of the critical bins. Consider a bin $B$ in \Gtwo. It contains precisely one 2-item having some weight $\geq1/2$. Assume that we have sorted the bins in \Gtwo\ in order of decreasing weight of its 2-item, i.e., we have an ordering of the bins $B_1,\ldots,B_{|\Gtwo|}$ such that the weight if the 2-item in $B_i$ is $w_i$, $1\leq i\leq|\Gtwo|$ and $w_i\geq w_{i+1}$ for $1\leq i\leq|\Gtwo|-1$; see Figure~\ref{criticalfig}. Let $u_i=\ld{B_i}-w_i$ be the weight of the small items in~$B_i$ and let $S_2=\sum_{1\leq i\leq|\Gtwo|}u_i$ be the total load of the small items covering the bins in~\Gtwo. It is clear that $S_2\geq(|\Gtwo|-m)(1-w_m)$ for arbitrary choice of $m\leq|\Gtwo|$ since $B_{m+1},\ldots,B_{|\Gtwo|}$, each contains at least $1-w_m$ amount of small items. 
\begin{figure*}
\centering
\input{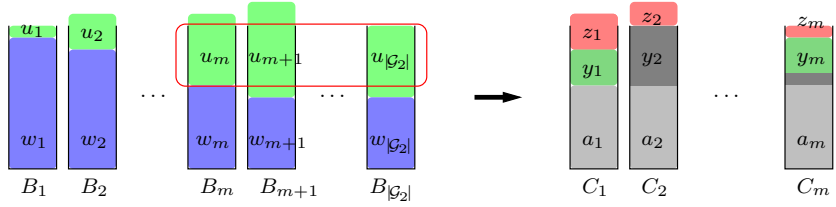}
\caption{\label{criticalfig}The critical bins and their relationship to the \Gtwo-bins in the optimal covering. In the \Gtwo-bins, blue are 2-items and light green are the small items. In the critical bins, red represents the last small item placed in the bin, dark green are the remaining small items, grey items are the 2-items, and  dark grey represents the overlap between the virtual and actual load of the~2-item.}
\end{figure*}

Let $C_i$, $1\leq i\leq m$, be the critical bins opened by our strategy in the order they are constructed. Let $a_i\geq x_m\geq w_m$ denote the weight of the 2-item in $C_i$, let $z_i$ be the weight of the last small item placed in $C_i$ by our strategy, and let $y_i=\ld{C_i}-a_i-z_i$ be the weight of the remaining small items in~$C_i$.  By construction, $y_i\leq1-w_m$ for each $1\leq i\leq m$.

Let $\Itwo{m}\subseteq\{1,\ldots,m\}$ be the set of indices $i$ such that the last small element (of weight $z_i$) that was placed in critical bin $C_i$ was placed by the optimal solution in a bin in \Gtwo.  Similarly, let $\IS{m}=\{1,\ldots,m\}\setminus\Itwo{m}$ be the set of remaining indices.
The possible values of $m$ range between $0\leq m\leq m^+=\lfloor(|\Gtwo|-|\Itwo{m^+}|)/2\rfloor$, where $m^+$ is the largest integer such that $2m^+ +|\Itwo{m^+}|\leq|\Gtwo|$, since the strategy needs to guarantee that it can cover all the critical bins.

The oracle can ascertain the value $m^+$ by emulating the strategy over all possible integer values and it reveals the values $m=m^+=\lfloor(|\Gtwo|-|\Itwo{m^+}|)/2\rfloor$ and $x_m=x_{m^+}$, the size of the $m^{\rm th}$ largest item in the input sequence $\sigma$ to the strategy. The strategy constructs $m$ critical bins, $\lfloor(T_2-m)/2\rfloor$ 2-bins and some bins corresponding to the amount of unused small items, giving us
{\small%\vspace*{-1ex}
\begin{align*}
%\lefteqn{
\big|\DHab(\sigma)\big| 
&>
m + \bigg\lfloor\frac{T_2-m}{2}\bigg\rfloor 
+ \frac{2}{3}\bigg(S_S+S_{2}
%}
%\\
%&\hspace*{2em}
-\bigg(\sum_{i=1}^m y_i+z_i\bigg)\bigg) - \frac{2}{3}
\\
%\end{align*}
%\begin{align*}
&\geq
m + \frac{T_2-m}{2} -\frac{1}{2} + \frac{2}{3}\bigg(S_S + S_{2} -\sum_{i=1}^{m}(1-w_{m}) 
%\\
%&\hspace*{2em}
-\sum_{i=1}^{m}z_i \bigg) -  \frac{2}{3}
\\
&=
\frac{m}{2} + \frac{T_2}{2} + \frac{2}{3}\bigg(S_S-\!\!\!\!\!\!\sum_{i\in\IS{m}}\!\!\!\!\!z_i\bigg)
%\\
%&\hspace*{2em}
+ \frac{2}{3}\underbrace{\bigg(S_{2} -\sum_{i=1}^{m}(1-w_m)-\!\!\!\!\!\!\sum_{i\in\Itwo{m}}\!\!\!\!\!z_i\bigg)}_{\geq0}
 -\frac{7}{6}
\\
&\geq
\frac{m}{2} + \frac{T_2}{2} + \frac{2}{3}\bigg(S_S-\!\!\!\!\!\!\sum_{i\in\IS{m}}\!\!\!\!\!z_i\bigg) -\frac{7}{6}
%\\
%&\geq
\ \ \geq \ \
\frac{m}{2} + \frac{T_2}{2} + \frac{3}{5}\bigg(S_S-\!\!\!\!\!\!\sum_{i\in\IS{m}}\!\!\!\!\!z_i\bigg)  -\frac{7}{6}
\\
&>
\frac{m}{2} + \frac{T_2}{2} + \frac{3S_S}{5} - \frac{3|\IS{m}|}{10}  -\frac{7}{6}
%\\
%&=
\ \ = \ \
\frac{m}{5} + \frac{T_2}{2} + \frac{3S_S}{5} + \frac{3|\Itwo{m}|}{10} -\frac{7}{6}
\\
&=
\frac{\lfloor(|\Gtwo|-|\Itwo{m}|)/2\rfloor}{5} + \frac{2|\Gtwotwo| + |\Gtwo|}{2} + \frac{3S_S}{5} 
%\\
%&\hspace*{2em}
+ \frac{3|\Itwo{m}|}{10} -\frac{7}{6}
%\end{align*}
%\begin{align*}
%&\geq
\ \ \geq \ \
|\Gtwotwo| + \frac{3|\Gtwo|}{5} + \frac{3|\GS|}{5} -\frac{19}{15}
\\
&\geq
\frac{3}{5}\big|\OPTs\big| -\frac{19}{15}
\end{align*}%
}%
bins, by applying Equality~(\ref{eq:twoeq}) in the last inequality, while using that each $z_i<1/2$, that $S_S\geq\sum_{i\in\IS{m}}z_i$, that critical bin $C_i$ can be covered by a 2-item of size at least $w_m$ plus the small items from a bin among $B_{m+1},\ldots,B_{|\Gtwo|}$ and one extra small item from a bin among $\big\{B_{m+1},\ldots,B_{|\Gtwo|}\big\}\cup\GS$; see Figure~\ref{criticalfig}, and that $m=|\Itwo{m}|+|\IS{m}|$, for any $m$. The competitive ratio is the smallest coefficient of any of the terms corresponding to bin groups, since an adversary can ensure that the groups with larger coefficient contain no bins. This gives a competitive ratio of~$3/5=0.6$.
\end{proof}
We have presented the proof for two item types to illustrate the general idea of the next result. For two item types, the number of bin groups is three, whereas for four item types, the number of bin groups increases to 20, but the proof steps are exactly the same. For completeness we mention that $|\DHac(\sigma)|\geq9|\OPTs|/14 - 97/42$, where $9/14\approx0.64285\ldots$.

\begin{lemma}\label{lem:twothirds}
Assume that the strategy \DHad\ has access to the exact values of $m$ and $x_m$, then it has asymptotic competitive ratio
$$
\big|\DHad(\sigma)\big|\geq\frac{2}{3}\big|\OPTs\big| - \frac{173}{60}
$$
for serving any sequence $\sigma$ of size~$n$.
\end{lemma}
\begin{proof}
The number of $t$-items, for $t=2$,~$3$, and~$4$, in the instance is
\begin{align}
T_2&=|\Gt{2}| + 2|\Gt{22}| + |\Gt{23}| + |\Gt{24}| + |\Gt{233}|  + \mbox{}
%\nonumber
%\\&\qquad
 + |\Gt{234}| + |\Gt{244}|,\label{eq:t2eq}
 \\
T_3&=|\Gt{3}| + |\Gt{23}| + 2|\Gt{33}| + |\Gt{34}| + 2|\Gt{233}| + \mbox{}
%\nonumber
%\\&\qquad
 + |\Gt{234}| + 3|\Gt{333}| + 2|\Gt{334}| + |\Gt{344}| + \mbox{}
\nonumber\\
&\qquad
 + 2|\Gt{3344}|, + |\Gt{3444}|,\label{eq:t3eq}
 \\
T_4&=|\Gt{4}| + |\Gt{24}| + |\Gt{34}| + 2|\Gt{44}| + |\Gt{234}| + \mbox{}
%\nonumber
%\\&\qquad
 + 2|\Gt{244}| + |\Gt{334}| + 2|\Gt{344}| + 3|\Gt{444}| + \mbox{}
\nonumber\\
&\qquad
 + 2|\Gt{3344}| + 3|\Gt{3444}| + 4|\Gt{4444}|.\label{eq:t4eq}
\end{align}

For each non-easy bin group $\Gtwo,\ldots,\Gt{444}$ (there are eight of them), let $S_{t_1\cdots t_4}$ denote the weight of the small items that the optimum solution packs in the bins of group~\Gt{t_1\cdots t_4}. 
In addition, we denote by $S_S=\sum_{B\in\GS}\ld{B}$ the total load of the small items covering the bins in~$\GS$.

As in the previous proof, we first consider some arbitrary set of covered bins \G, where each bin only contains small items. Assume that these bins have a total load of $S=\sum_{B\in\G}\ld{B}\geq|\G|$ and that the input sequence restricted to these small items is $\sigma_S$. From Lemma~\ref{lem:dnf} we have that 
%\vspace*{-0.75ex}
\begin{equation}
\big|\DNF(\sigma_S)\big|  > \frac{4}{5}S - \frac{4}{5} \geq \frac{4}{5}\big|\G\big| - \frac{4}{5}
%\vspace*{-0.75ex}
\end{equation}

We can analyze the competitive ratio of the critical bins exactly as in the previous proof. First consider a decreasing ordering of the bins $B_1,\ldots,B_{|\Gtwo|}$ in \Gtwo\ by the weight of their 2-item, $w_i$. We let $u_i=\ld{B_i}-w_i$ be the weight of the small items in~$B_i$, whereby $S_2\geq\big(|\Gtwo|-m\big)\cdot\big(1-w_m\big)$ for arbitrary choice of $m\leq|\Gtwo|$.
%since $B_{m+1},\ldots,B_{|\Gtwo|}$, each contains at least $1-w_m$ amount of small items; see Figure~\ref{criticalfig}.
%
The critical bins, $C_i$, $1\leq i\leq m$, each contains one 2-item of weight $a_i$, a small item of weight $z_i$ that was the last small item placed in $C_i$ by our strategy, and small items to the weight of $y_i=\ld{C_i}-a_i-z_i$.  Again, %by construction, 
$y_i\leq1-w_m$, for each $1\leq i\leq m$.

Consider next the gap bins in the optimal solution. These are the bins in groups \Gt{3}, \Gt{4}, \Gt{34}, and \Gt{44}. Each bin in these groups is guaranteed to have small items to the amount of at least $1/2$, $2/3$, $1/6$, and $1/3$, respectively. Thus, for each of those groups we have $S_3\geq|\Gt{3}|/2$, $S_4\geq2|\Gt{4}|/3$, $S_{34}\geq|\Gt{34}|/6$, and $S_{44}\geq|\Gt{44}|/3$.

For each group of non-easy bins $\Gtwo,\ldots,\Gt{444}$, let $\It{m}{t_1\cdots t_4}\subseteq\{1,\ldots,m\}$ be the set of indices $i$ such that the last small element (of weight $z_i$) that was placed in critical bin $C_i$ was placed by the optimal solution in a bin from bin group~\Gt{t_1\cdots t_4}. %Easy bins are assumed, without loss of generality, to not contain any small items. 
Also, let $\IS{m}=\big\{1,\ldots,m\big\}\setminus\big(\bigcup_{t_1\cdots t_4\not\in{\it Easy}}\It{m}{t_1\cdots t_4}\big)$ be the set of remaining indices.
As before, the possible values of $m$ range between $0\leq m\leq m^+=\big\lfloor(|\Gtwo|-|\Itwo{m^+}|)/2\big\rfloor$, where $m^+$ is the largest integer such that $2m^+ +|\Itwo{m^+}|\leq|\Gtwo|$, since the strategy needs to guarantee that it can cover all the critical bins.
The oracle ascertains the maximum integer $m^+$ and reveals the values  $m=m^+=\big\lfloor(|\Gtwo|-|\Itwo{m^+}|)/2\big\rfloor$ and $x_m=x_{m^+}$, %the $m^{\rm th}$ largest item in the input sequence $\sigma$, 
so our strategy constructs $m$ critical bins, $\big\lfloor(T_2-m)/2\big\rfloor$ $2$-bins, $\big\lfloor T_3/3\big\rfloor$ $3$-bins, $\big\lfloor T_4/4\big\rfloor$ $4$-bins, and bins corresponding to the amount of unused small items, giving us
{\small%\vspace*{-2.75ex}
\begin{align*}
%\lefteqn{\!\!\!\!\!
\big|\DHad(\sigma)\big| 
>&\ 
m + \left\lfloor\frac{T_2-m}{2}\right\rfloor + \left\lfloor\frac{T_3}{3}\right\rfloor + \left\lfloor\frac{T_4}{4}\right\rfloor
%} 
%&
%\\
%&%\hspace*{-20em}
+\frac{4}{5}\bigg(S_S+S_2+S_3+S_{4}+S_{34}+S_{44}
%\\
%&%\hspace*{-20em}
-\bigg(\sum_{i=1}^m y_i+z_i\bigg)\bigg) - \frac{4}{5}
\\
\geq&\
\frac{m}{2} + \frac{T_2}{2} + \frac{T_3}{3}  + \frac{T_4}{4} + \frac{4}{5}\bigg(S_S-\!\!\!\!\!\!\sum_{i\in\IS{m}}\!\!\!\!\!z_i\bigg) 
%\\
%&
+ \frac{4}{5}\bigg(S_{3}-\!\!\!\!\!\!\sum_{i\in\It{m}{3}}\!\!\!\!\!z_i\bigg) 
+ \frac{4}{5}\bigg(S_{4}-\!\!\!\!\!\!\sum_{i\in\It{m}{4}}\!\!\!\!\!z_i\bigg)
\\
&
+ \frac{4}{5}\bigg(S_{34}-\!\!\!\!\!\!\sum_{i\in\It{m}{34}}\!\!\!\!\!z_i\bigg) + \frac{4}{5}\bigg(S_{44}-\!\!\!\!\!\!\sum_{i\in\It{m}{44}}\!\!\!\!\!z_i\bigg)
%\\
%&
+ \frac{4}{5} \underbrace{\bigg(S_{2} -\sum_{i=1}^{m}(1-w_m)-\!\!\!\!\!\!\sum_{i\in\It{m}{2}}\!\!\!\!\!z_i\bigg)}_{\geq0}  -\frac{163}{60}
\nonumber\\
\end{align*}
\begin{align*}
\geq&\
\frac{m}{2} + \frac{T_2}{2} + \frac{T_3}{3}  + \frac{T_4}{4} + \frac{2}{3}\bigg(S_S-\!\!\!\!\!\!\sum_{i\in\IS{m}}\!\!\!\!\!z_i\bigg) 
%\nonumber\\
%&
+ \frac{2}{3}\bigg(S_{3}-\!\!\!\!\!\!\sum_{i\in\It{m}{3}}\!\!\!\!\!z_i\bigg) + \frac{5}{8}\bigg(S_{4}-\!\!\!\!\!\!\sum_{i\in\It{m}{4}}\!\!\!\!\!z_i\bigg)
\nonumber\\
&
+ \frac{1}{2}\bigg(S_{34}-\!\!\!\!\!\!\sum_{i\in\It{m}{34}}\!\!\!\!\!z_i\bigg) + \frac{1}{2}\bigg(S_{44}-\!\!\!\!\!\!\sum_{i\in\It{m}{44}}\!\!\!\!\!z_i\bigg) -\frac{163}{60}
\nonumber\\
\geq&\
\frac{m}{2} + \frac{T_2}{2} + \frac{T_3}{3} + \frac{T_4}{4} + \frac{2S_S}{3} - \frac{|\IS{m}|}{6} 
%\nonumber\\
%&
+ \frac{2S_{3}}{3} - \frac{|\It{m}{3}|}{6} 
+ \frac{5S_{4}}{8} - \frac{5|\It{m}{4}|}{32} 
\nonumber\\
&
+ \frac{S_{34}}{2} - \frac{|\It{m}{34}|}{8}
+ \frac{S_{44}}{2} - \frac{|\It{m}{44}|}{8} - \frac{163}{60}
\nonumber\\
\geq&\
\frac{m}{2} + \frac{T_2}{2} + \frac{T_3}{3} + \frac{T_4}{4} + \frac{2S_S}{3} - \frac{m}{6} + \frac{|\It{m}{2}|}{6} 
%\nonumber\\
%&
+ \frac{|\It{m}{3}|}{6} + \frac{|\It{m}{4}|}{6} + \frac{|\It{m}{34}|}{6} + \frac{|\It{m}{44}|}{6} 
\nonumber\\
&
+ \frac{2S_{3}}{3} - \frac{|\It{m}{3}|}{6}
+ \frac{5S_{4}}{8} - \frac{5|\It{m}{4}|}{32}
%\nonumber\\
%&
+ \frac{S_{34}}{2} - \frac{|\It{m}{34}|}{8} + \frac{S_{44}}{2} - \frac{|\It{m}{44}|}{8} 
- \frac{163}{60}
\nonumber\\
>&\
\frac{m}{3} + \frac{T_2}{2} + \frac{T_3}{3} + \frac{T_4}{4} + \frac{|\It{m}{2}|}{6} + \frac{2|\GS|}{3} + \frac{|\Gt{3}|}{3} 
%\nonumber\\
%&
+ \frac{5|\Gt{4}|}{12} + \frac{|\Gt{34}|}{12}  + \frac{|\Gt{44}|}{6} -\frac{163}{60}
\nonumber\\
=&\
\frac{2|\Gt{2}|}{3} + \frac{2|\Gt{3}|}{3} + \frac{2|\Gt{4}|}{3} + \frac{2|\Gt{33}|}{3} + \frac{2|\Gt{34}|}{3} + \frac{2|\Gt{44}|}{3} 
%\nonumber\\
%&
+ \frac{2|\GS|}{3} + \frac{3|\Gt{24}|}{4} + \frac{3|\Gt{444}|}{4} + \frac{5|\Gt{23}|}{6} + \frac{5|\Gt{344}|}{6} 
\nonumber\\
&
+ \frac{11|\Gt{334}|}{12} + |\Gt{22}| + |\Gt{244}| + |\Gt{333}| + |\Gt{4444}| 
%\nonumber\\
%&
+ \frac{13|\Gt{234}|}{12} + \frac{13|\Gt{3444}|}{12} + \frac{7|\Gt{233}|}{6} + \frac{7|\Gt{3344}|}{6} -\frac{173}{60} 
\nonumber\\
\geq&\
%\hspace*{2em} \geq \hspace*{2em}
\frac{2}{3}\big|\OPTs\big| -\frac{173}{60} 
\end{align*}%
}%
bins, by applying Equalities~(\ref{eq:t2eq})--(\ref{eq:t4eq}) in the second to last inequality above, while using that each $z_i<1/4$, that $S_S\geq\sum_{i\in\IS{m}}z_i$ and $S_{t_1\cdots t_4}\geq\sum_{i\in\It{m}{t_1\cdots t_4}}z_i$, for each bin group \Gt{t_1\cdots t_4}, that critical bin $C_i$ can be covered by a large item of size at least $w_m$ plus the small items from a bin among the last bins $B_{m+1},\ldots,B_{|\Gtwo|}$ in \Gtwo\ and one extra small item from a non-easy bin in the optimal solution; see Figure~\ref{criticalfig}, and that $m=|\IS{m}|+\sum_{t_1\cdots t_4}|\It{m}{t_1\cdots t_4}|$, for any $m$. The competitive ratio is the smallest coefficient of any of the terms corresponding to bin groups, since an adversary can ensure that the groups with larger coefficient contain no bins. This gives a competitive ratio of~$2/3\approx0.6666\ldots$.
\end{proof}

The two advice values $m\leq n$ and $x_m$ can be represented by $O(\log n)$ bits and $O(b)$ bits respectively, where $b$ is the number of bits required to represent the integer denominator of the rational value $x_m$, since $x_m<1$. We have the following immediate theorem.
\begin{theorem}\label{thm:twothirds}
The strategy \DHad\ receives $O(b+\log n)$ bits of advice and has asymptotic competitive ratio
$$
\big|\DHad(\sigma)\big|\geq\frac{2}{3}\big|\OPTs\big| - \frac{173}{60}
$$
for serving any sequence $\sigma$ of size~$n$, where $b$ is the number of bits required to represent any rational value in~$\sigma$.
\end{theorem}

By approximating the advice value $m$ using $O(\log\log n)$ of the most significant bits (using a modified self-delimiting encoding), a slight variation of the analysis above shows that a variant of \DHad\ that receives  $O(b+\log\log n)$ bits of advice still has asymptotic competitive ratio $2/3-1/O(\log n)$, thus the approximate value of $m$ has very little impact on the competitive ratio.

\subsection{Tightness}
One could venture to think that strategy \DHak, for $k>4$, would give improved competitive ratio, or even that extending the strategy with more sets of critical bins could improve it further. However, this is not possible, since an adversary can simply provide an instance where all bin groups except \Gtwo\ in an optimal solution are empty. Thus, the instance consists of only $2$-items and small items. Any critical bin-based strategy must solve this instance and chooses some value for $m$, the number of critical bins to open. Even if the adversary provides all the small items first and the 2-items last, the strategy will cover $m+\big\lfloor(|\Gtwo|-m)/2\big\rfloor$ bins as long as the strategy can guarantee that all $m$ critical bins are covered with one 2-item and some small items. Since the index set is $\It{m}{2}=\{1,\ldots,m\}$, for all $m$, the maximum occurs for $m=\lfloor|\Gtwo|/3\rfloor$. The strategy covers at most
{\small
\begin{align*}
m+\left\lfloor\frac{|\Gtwo|-m}{2}\right\rfloor 
&\ \ = \ \
\left\lfloor\frac{|\Gtwo|}{3}\right\rfloor + \left\lfloor\frac{|\Gtwo|-\lfloor|\Gtwo|/3\rfloor}{2}\right\rfloor
\\
&\hspace*{-5em}
\leq\ \
\frac{|\Gtwo|}{3} + \frac{|\Gtwo|}{3} + \frac{1}{3} 
%\\
%&
\ \ \leq \ \
\frac{2}{3}\big|\OPTs\big| + \frac{1}{3}
\end{align*}
}%
bins, proving that our analysis in Lemma~\ref{lem:twothirds} is asymptotically tight.

}

\end{document}